\documentclass[runningheads]{llncs}

\usepackage{graphicx}
\usepackage{cite}
\usepackage{graphicx}
\usepackage{textcomp}
\usepackage{xcolor}
\usepackage[ruled,vlined]{algorithm2e}
\usepackage{algpseudocode}
\usepackage{amsmath}
\usepackage{amsfonts}
\usepackage{bbm}
\usepackage{mathtools}
\usepackage{diagbox}
\usepackage{array}
\usepackage{amsmath}
\usepackage{amssymb}
\usepackage{mathtools}
\usepackage[english]{babel}

\usepackage{breqn}
\usepackage{amsthm}
\usepackage{thmtools}
\usepackage{enumitem}
\usepackage{physics}
\usepackage{thm-restate}
\usepackage[english]{babel}
\usepackage{url}
\theoremstyle{definition}
\theoremstyle{remark}

\newcommand{\comment}[1]{}
\usepackage[title]{appendix}
\begin{document}
\title{Establishing the Price of Privacy in Federated Data Trading}
\author{Kangsoo Jung\inst{} \and
Sayan Biswas\inst{} \and
Catuscia Palamidessi\inst{}}
\authorrunning{}
\institute{Inria  and Ecole Polytechnique, France\\
\email{\{gangsoo.zeong,sayan.biswas\}@inria.fr}\\
\email{catuscia@lix.polytechnique.fr}
\url{}}

\maketitle               % typeset the header of the contribution
\begin{abstract}
Personal data is  becoming one of the most essential resources in today's information-based society. Accordingly, there is a growing interest in data markets, which operate data trading services between data providers and data consumers. One issue the data markets have to address is that of the  potential threats to privacy. Usually some kind of protection must be provided, which generally comes to the detriment of utility. A  correct pricing mechanism for private data should therefore depend on the level of privacy. In this paper, we propose a model of data federation in which data providers, who are, generally, less influential on the market than data consumers, form a coalition for trading their data, simultaneously shielding against privacy threats by means of differential privacy. Additionally, we propose a technique to price private data, and an revenue-distribution mechanism to distribute the revenue fairly in such federation data trading environments. Our model also motivates the data providers to cooperate with their respective federations, facilitating a fair and swift private data trading process. We  validate our result  through various experiments, showing  that the proposed methods provide   benefits to both data providers and consumers.

\end{abstract}
\section{Introduction}

%Nowadays personal data is regarded as the crude oil of twenty-first century, and it plays a key role in progressing towards an information-based society. 
The use of data analytics 
%-driven decision-making 
is growing, as it plays a crucial role 
in making decisions and identifying social and economical strategies. Not all data, however, are equally useful, and  
the availability of accurate data is crucial for obtaining high-quality  analytics.
%to contribute towards the success of various organizations. 
In line with this trend,  data are considered an asset and   commercialized, and  data markets, such as Datacoup\cite{Datacoup84:online} and Liveen\cite{LiveenBl14:online},
are on the rise. 

 Unlike traditional  data brokers, data markets provide  a  direct data trading service between data providers and data consumers. Through data markets, data providers can be informed of the value of their private data, and data consumers can  collect and process personal data directly at  reduced costs, as intermediate entities are not needed in this model. 

Two important issues that need to be addressed for the success of such data markets are (a) the prevention of privacy violation, and (b) an appropriate pricing mechanism for personal data. Data owners are increasingly aware of the privacy risks, and are less and less inclined to expose their sensitive data without proper guarantees. If the data market cannot be trusted concerning the   protection of the sensitive information, the data providers will not be willing to trade their data. 
For example, Cambridge Analytica collected millions of Facebook users' profiles under the pretext of using them for academic purposes, while in reality they used this information to  influence  the 2016 US presidential election \cite{hinds2020wouldn}. When media outlets broke news of Cambridge Analytica's business practices,  many Facebook users felt upset  about the misuse of their data and left Facebook. 

Differential privacy \cite{dwork2014algorithmic}  can prevent exposure of personal information while preserving statistical utility, hence it is a good candidate  to protect privacy in the data market. Another benefit of  differential privacy is that it  provides  a metric, i.e., the  parameter
$\epsilon$, which represents the amount of obfuscation, and therefore the level of privacy and utility of the sanitized data. Hence $\epsilon$ can be used directly to establish the price of personal data as a function of the level of privacy protection desired by an individual.

We envision  a data trading framework in which groups of data providers ally to form federations in order 
to increase their  bargaining  power, following the traditional model of trade unions. 
At the same time, federations guarantee that the members respect their engagement concerning the trade. 
Another important aspect  of the   federation  is that the value of the collection of all data is usually different from the sum of the values of all members' data. It  could be larger, for instance because  the accuracy of the statistical analyses increases with the size of the dataset, or could be smaller, for instance because of some discount offered by the federation. 
Data consumers are supposed to make a collective deal with a federation rather than with the individual data providers, and, from their perspective,  
this approach can be more reliable and efficient than dealing with individuals. Thus, data trading through federations can benefit  both  parties.

Given such a scenario, two questions are in order: 
\begin{enumerate}
\item{ How is the price of data determined in a federation environment?}
\item { How does the federation fairly distribute the  earnings to its members?}
\end{enumerate}

In this paper, we consider these issues, and we provide the following contributions:  

\begin{enumerate}
\item{ We propose a method to determine the price of collective data based on the differential privacy
metric.}
\item {We propose a distribution model based on game theory. More precisely, we borrow the
notion of Shapley value~\cite{winter2002shapley, roth1988shapley} from the theory of cooperative games. This is a method to determine the contribution of each participant to the payoff, and we will use it 
to ensure that each member of the federation receives a compensation according to his contribution.}
\end{enumerate}

The paper is organized as follows:  Sections 2 
 recalls some basic notions about differential privacy and Shapley values.  Section  3  summarizes related works. Section 4  describes the federation-based data trading and our proposal for the distribution of the earnings.  Section 5 validates the proposed technique through  experiments. Section 6    concludes and discusses  potential directions for future work.

\section{Preliminaries}
In this section, we recall the basics about differential privacy and Shapley values.
\subsection{Differential privacy}
Differential privacy (DP) is a method  to ensure  privacy
on  datasets  based on obfuscating the answers to queries. It is parametrized by   $\epsilon\in \mathbb{R}^+$, that represents the level of privacy. 
We recall that  two datasets  $D_1$ and $D_2$ are 
neighboring if they differ by only one record.

\begin{definition}[Differential privacy\text{\cite{dwork2014algorithmic}}]
\label{dp}
A randomized function $\mathcal{R}$ provides \emph{$\epsilon$-differential privacy} if for all neighboring datasets   $D_1$ and $D_2$  and all $ S \subseteq$ Range($\mathcal{R}$), we have
\begin{equation*}
\mathbb{P}[\mathcal{R}(D_1) \in S] \leq e^{\epsilon} \times  \mathbb{P}[\mathcal{R}(D_2) \in S]
\end{equation*}
\end{definition}

For example, if we have $\mathbb{D}$ as the space of all datasets, and some $m\in \mathbb{N}$, then the randomized function $\mathcal{R}:\mathbb{D}\mapsto\mathbb{R}^m$ could be such that $\mathcal{R}(D) = \mathcal{Q}(D)+X$, where $\mathcal{Q}$ is a statistical query function executed on $D$, such as the counting or histogram query, and $X$ is some added noise to the true query response. For $\Delta_{\mathcal{Q}}=\max\limits_{D,D'\in\mathbb{D}}|\mathcal{Q}(D)-\mathcal{Q}(D')|$, if $X\sim\text{Lap}(0,\frac{\Delta_{\mathcal{Q}}}{\epsilon})$, $\mathcal{R}$ will guarantee $\epsilon$-DP.

DP is typically implemented by adding controlled  random noise to the true answer to the query before reporting the result.  $\epsilon$ is a positive real number parameter, and the value of $\epsilon$ affects the amount of privacy, which decreases as  $\epsilon$ increases. For simplicity of discussion, we focus on the non-interactive and pure $\epsilon$-differential privacy.

Recently, a local variant of differential privacy (LDP), in which the data owner directly obfuscate their data, has been proposed\cite{erlingsson2014rappor}. This variant considers the individual data points (or records), rather than queries on datasets.
Its definition  is as follows:
\begin{definition}[Local differential privacy\cite{erlingsson2014rappor}]
A randomized function  $\mathcal{R}$ satisfies $\epsilon$-\emph{local differential privacy} if, for all pairs of individual data $x$ and $x'$,  and for any subset $S \subseteq \text{Range}(\mathcal{R})$, we have
\begin{equation*}
   \mathbb{P}[\mathcal{R}(x) \in S] \leq e^{\epsilon} \cdot  \mathbb{P}[\mathcal{R}(x')] \in S,   
\end{equation*}
\label{ldp}
\end{definition}

When the domain of data points is finite, one of the simplest and  most used mechanisms for LDP is  $k$RR\cite{kairouz2016discrete}. 
In this paper, we assume that all data providers use this mechanism to obfuscate their data.
\begin{definition}[$k$RR Mechanism \cite{kairouz2016discrete}]
Let $\mathcal{X}$ be an alphabet of size $k < \infty$. For a given privacy parameter $\epsilon$, and given 
an input $x\in \mathcal{X}$, the $k$RR mechanism  
returns $y\in \mathcal{X}$ with probability:
\[\mathbb{P}(y|x)\;\; = \;\;\frac{1}{k-1+e^{\epsilon}} \begin{cases}
e^{\epsilon}, & \mbox{if } y=x \\
1, & \mbox{if } y \neq x
\end{cases}\]
\label{krr}
\end{definition}

\subsection{Shapley value}

When participating in data trading through a federation, \emph{Pareto efficiency} and \emph{symmetry} are the important properties for the intra-federation earning distribution. Pareto efficiency means that at the end of the distribution process, no change can be made without making participants worse off. Symmetry means that all players who make the same contribution must receive the same share. Obviously, the  share should vary according to the member's contribution  to the collective data.

The Shapley value\cite{winter2002shapley, roth1988shapley} is a concept from game theory named in honor of Lloyd Shapley, who introduced it. Thanks to this achievement, Shapley won the Nobel Prize in Economics in 2012.  
The Shapley value applies to cooperative games, and it 
is a method  to distribute the total gain that satisfy Pareto efficiency, symmetry, and differential distribution according to a player's contribution. 
Thus, all participants have the advantage of being fairly incentivized.
The solution based on the Shapley value 
is unique. Due to these properties, the Shapley value is regarded as an excellent approach to design a distribution method. 

Let $N=\{1,\ldots,n\}$ be a set of players involved in a cooperative game and $M\in\mathbb{R}^+$ be a financial revenue from the data consumer. Let $v: 2^N\mapsto\mathbb{R}^+$ be the characteristic function, mapping each subset  $S\subseteq N$ to  the total expected sum of payoffs the members of $S$ can obtain by cooperation. (i.e., $v(S)$   is the total collective payoff of the players in $S$). According to the Shapley value, the benefit received by player $i$ in the cooperative game is given follows: 

\[ \psi_i(v, M) \;\; =\;  \sum_{S \subseteq N\setminus\{i\}} \frac{|S|!\times(n-|S|-1)!}{n!}(v(S \cup \{i\} )-v(S)) \]

We observe that $v(A) > v(B)$ for any subsets $B \subset A \subseteq N$, and hence, $v(S \cup \{i\} )-v(S)$ is positive. We call this quantity the \emph{marginal contribution} of player $i$ in a given subset $S$. Note that $\psi_i(v, M)$ is the expected marginal contribution of player $i$ over all subsets $S \subseteq N$. 

In this paper, we use the Shapley value to distribute the earnings according to the contributions of the data providers in the federations.
 
\section{Related works}
Data markets, such as Datacoup\cite{Datacoup84:online} and Liveen\cite{LiveenBl14:online}, need to  provide privacy protection in order to encourage the data owners to participate. 
One of the key questions is how to appropriately price data obfuscated by
a privacy-protection mechanism. 
When we use differential privacy, the accuracy of data depends on the value of the noise parameter $\epsilon$, which determines the privacy-utility trade-off. Thus, this question is linked to the problem of how to establish the value of  $\epsilon$. Researchers have debated how to choose this value since the introduction of differential privacy, and there have been several proposals  \cite{tang2017privacy, lee2011much,domingo2015t,holohan2017k}. In particular, \cite{lee2011much} showed that the privacy protection level by an arbitrary $\epsilon$ can be infringed by inference attacks, and it proposed a method for setting $\epsilon$ based on the posterior belief.  \cite{domingo2015t} considered the relation between differential privacy and t-closeness, a notion of group privacy which prescribes that the earth movers distance between the distribution in any  group $E$ and the distribution in the whole dataset does not exceed the threshold $t$, and showed that  both $\epsilon$-differential privacy and t-closeness are satisfied  when the $t=\max_E\frac{|E|}{N}\left(1+\frac{N-|E|-1}{|E|})e^{\epsilon}\right)$ where  $N$ is the number of records of the database. 

Several other works have studied how to price the data according to the value of $\epsilon$ \cite{hsu2014differential, ghosh2011selling, roth2012buying, fleischer2012approximately, nget2017balance,li2014theory,zhang2021differential,jung2019privacy}. The purpose of these studies is to determine the price and value of the $\epsilon$ according to the data consumer’s budget, accuracy requirement of information, the privacy preference of the data provider, and the relevance of the data. In particular, the study in \cite{zhang2021differential} assumed a dynamic data market and proposed an incentive mechanism for data owners to truthfully report their privacy preferences. In \cite{nget2017balance}, the authors proposed a framework to find the balance between financial incentive and privacy in personal data markets where data owners sell their own data, and suggested the main principles to achieve reasonable data trading. Ghosh and Roth  \cite{ghosh2011selling} proposed a pricing mechanism based on auctions that maximizes the data accuracy under the budget constraint or minimizes the budget for the fixed data accuracy requirement, where data is privatized with differential privacy.

Our study differs from previous work in that, unlike the existing approaches assuming a one-to-one data trading between data consumers and  providers, we consider trades between a data consumer and a federation of data providers.   In such a federated  environment, the questions are  (a) how to determine the price of the collective  data according to the privacy preferences of each member, and (b) how to determine the individuals' contribution
to the overall data value, in order to receive a                                                                                                                         share of the earnings  accordingly.

In this paper, we estimate the value of $\epsilon$ for the  $k$RR mechanism \cite{kairouz2016discrete}, and we fairly distribute the earnings to the members of the federations  using the Shapley value. We propose a valuation function that fits the characteristics of differential privacy. For example, increasing value of $\epsilon$ does not infinitely increase the price  (we will elaborate on this in section 4). Furthermore, we characterize the conditions required for setting up the earning distribution schemes.

\section{Differentially Private Data Trading Mechanism}
\subsection{Mechanism outline}
    \paragraph{Overview: }We focus on an environment with multiple federations of data providers and one data consumer who interacts with the federations in order to obtain information (data obfuscated using $k$RR mechanism with varying values of $\epsilon$) in exchange of financial revenues. We assume that federations and  consumer are aware that the data providers use $k$RR mechanism, independently and with their desired privacy level (which can differ  from provider to provider). Our method provides a sensible way of splitting the earnings using the Shapley value. In addition, it also motivates an individual to cooperate with the federation she is a part of, and penalises   intentional and recurring non-cooperation. 
    \paragraph{Notations and set-up: } Let $\mathcal{F}=\{F_1,\ldots, F_k\}$ be a set of $k$ federations of data providers, where each federation $F_i$ has $n_{F_i}$ members for each $i \in \{1,\ldots,k\}$. For a federation $F\in\mathcal{F}$, let its members be denoted by $F=\{p^F_1,\ldots,p^F_{n_F}\}$. And finally, for every federation $F$, let $p^F_* \in F$ be an elected representative of $F$ interacting with the data consumer. This approach to communication benefits both the data consumer and the data providers because (a) the data consumer minimizes her communication cost  by interacting 
    with just one representative of the federation, and (b) the reduced communication induces an additional layer of privacy.

        We assume that each member $p$
        of a federation $F$ has  a maximum privacy threshold $\epsilon^T_p$ with which she, independently, obfuscates her data using the $k$RR mechanism. We also assume that $p$ has $d_p$ data points to potentially report. 
        
        We know from Equation (18) of \cite{Ehab2021privacy} that if there are $m$ data providers reporting $d_1,\ldots,d_m$ data points, independently  privatizing them using $k$RR mechanism with the privacy parameters $\epsilon_1,\ldots, \epsilon_m$, the federated data of all the $m$ providers also follow a $k$RR mechanism with the privacy parameter being:
        \[
        \epsilon=\ln{\left(\frac{\sum_{i=1}^m d_i}{\sum_{i=1}^m \;\frac{d_i}{k-1 + e^{\epsilon_i}}}\; +1-k\right)}.
        \]
        We call the quantity $d_p\epsilon^T_p$  the \emph{information limit} of data provider $p\in F$, and 
        \begin{equation}\label{eq:maxinfothreshold}
            \epsilon^T_F=\ln{\left(\frac{\sum_{p\in F} d_p}{\sum_{p\in F} \;\frac{d_p}{k-1 + e^{\epsilon^T_p}}}\; +1-k\right)}
        \end{equation}
        the \emph{maximum information threshold of the federation} $F$.
        
        We now introduce the concept of 
        \emph{valuation function} 
        $f(.)$, that maps financial revenues to  information, representing the amount of information to be obtained for a given price. It is reasonable to require that $f(.)$ is strictly monotonically increasing and continuous. 
        In this work we focus on the effect on the privacy parameter, hence we 
        regard the collection of data points as a constant, and assume that only $\epsilon$ can vary. We will call $f(.)$  the \emph{privacy valuation function}.

    \begin{definition}[Privacy valuation function]
        A function $f: \mathbb{R^+}\mapsto\mathbb{R^+}$ is a \emph{privacy valuation function} if $f(.)$ is strictly monotonically increasing and continuous.
        \label{privvalfn}
    \end{definition}
 
    As $f(.)$ is strictly monotonically increasing and continuous, it is also invertible. We denote the inverse of $f(.)$ as $f^{-1}(.)$, where $f^{-1}: \mathbb{R^+}\mapsto\mathbb{R^+}$, maps a certain privacy parameter $\epsilon$ to the financial revenue evaluated with selling data privatized using $k$RR mechanism with $\epsilon$ as the privacy parameter.
    
    As $f(.)$ is essentially determining the privacy parameter of a differentially private mechanism ($k$RR, in this case), it is reasonable to assume that $f(.)$ should be not only  increasing, but also  increasing exponentially for a linear increase of money. In fact, when $\epsilon$ is high, it hardly makes any difference to further increase its value. For example, when $\epsilon$ increases from 200 to 250, it practically makes no difference to the data as they were already practically no private. On the other hand, if we increase $\epsilon$ from 0 to 50, it creates a huge difference, conveying much  more information. Therefore, it makes sense to set  $f(.)$ to increase exponentially with a linear increase of the financial revenue.

    An example of a privacy valuation function that we consider in this paper is $f(M)=K_1(e^{K_2M}-1)$, taking the financial revenue $M\in\mathbb{R}^+$ as its argument, satisfying the reasonable assumptions of evaluating the differential privacy parameter that should be used to privatize the data in exchange of the financial revenue of $M$. Here the parameters $K_1\in \mathbb{R}^+$ and $K_2\in \mathbb{R}^+$ are decided by the data consumer according to her requirements.
    
\begin{figure}[htbp]
\centerline{\includegraphics[width=0.7\textwidth]{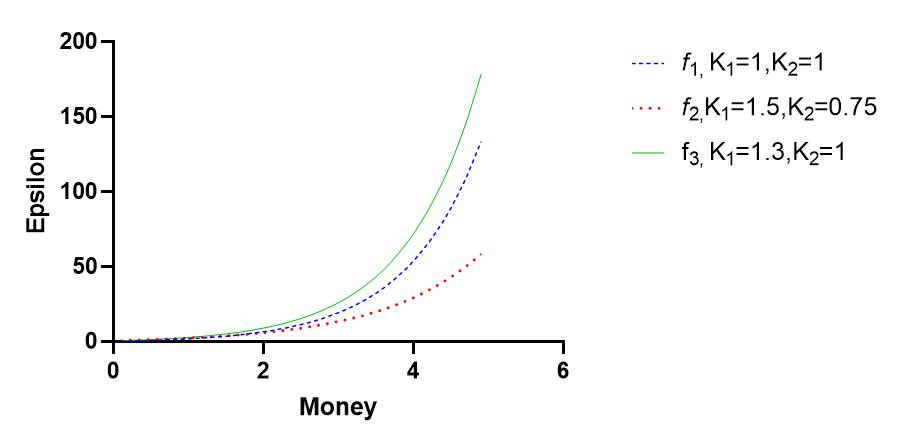}}
\caption{Some examples of the privacy valuation function $f(.)$ illustrated with different values of $K_1$ and $K_2$. The data consumer decides the values of the parameters $K_1$ and $K_2$ according to her requirement, and broadcasts the determined function to the federations.}
\label{fig1}
\end{figure}

    \paragraph{Finalizing and achieving the deal: }Before the private-data trading commences, the data consumer, $D$, truthfully broadcasts her financial budget, $\$B$, and a privacy-valuation function, $f(.)$, chosen by her to all the federations. At this stage, each federation computes their maximum privacy threshold. In particular, for a federation $F$ with members $F=\{p_1,\ldots,p_n\}$, and a representative $p_*$, $p_i$ reports $d_{p_i}$ and $\epsilon^T_{p_i}$ to $p_*$ for all $i \in \{1,\ldots, n\}$. $p_*$ computes the maximum information threshold of federation $F$, $\epsilon^T_F$, as given by \eqref{eq:maxinfothreshold}.

    At this point,  $p_*$ places a bid to $D$ to obtain \$$M$, which maximises the earning for $F$ under the constraint of their maximum privacy threshold and the maximum budget available from $D$, i.e., $p_*$ wishes to maximize $M$ within the limits $M\leq B$ and $f(M) \leq \epsilon^T_F$. Thus, $p_*$ bids for sending data privatized using the $k$RR mechanism with $\epsilon^T_F$ in exchange of $f^{-1}(\epsilon^T_F)$. 
    
    At the end of this bidding process by all the federations, $D$ ends up with $\epsilon=\{\epsilon^T_{F_1},\ldots,\epsilon^T_{F_k}\}$, the maximum privacy thresholds of all the federations. At this stage $D$ must ensure that $\sum_{i=1}^k f^{-1}(\epsilon^T_{F_i}) \leq B$, adhering to her financial budget. In all probability, $\sum_{i=1}^k f^{-1}(\epsilon^T_{F_i})$ is likely to exceed $B$ in a realistic setup. Here, $D$ needs a way to ``seal the deal'' with the federations staying within her financial budget, maximizing her information gain, i.e., maximizing $\sum_{i=1}^kd_{F_i}\epsilon_{F_i}$, where $d_{F_i}$ is the total number of data points obtained from the $i^{th}$ federation $F_i$, and $\epsilon_{F_i}$ is the overall privacy parameter of the $k$RR differential privacy with the combined data of all the members of $F_i$. 
    
    A way $D$ could   finalize the deal with the federations is by proposing to receive information obfuscated with $w^*\epsilon^T_{F_i}$ using $k$RR mechanism to $F_i$ $\forall i \in \{1,\ldots,k\}$, where \[w^*=\max\left\{w: \sum_{i\in\{1,\ldots,k\}}f^{-1}(w\epsilon^T_{F_i}) \leq B, w\in [0,1]\right\},\]
    i.e., proportional to every federation's maximum privacy threshold ensuring that the price to be paid to the federations is within $D$'s budget. Note that $w\in [0,1]$ guarantees that $w\epsilon^T_{F} \leq \epsilon^T_F$ for every federation $F$, making the proposed privacy parameter possible to achieve by every federation, as it's   within their respective maximum privacy thresholds. Let the combined privacy parameter for federation $F_i$ proposed by $D$ to successfully complete the deal be denoted by $\epsilon^P_{F_i}=w^*\epsilon^T_{F_i}$ $\forall i \in \{1,\ldots,k\}$ which is the privacy level \emph{promised} to be achieved by each federation to participate in the data market.
    
    The above method to scale down the maximum privacy parameters to propose a deal, maximizing $D$'s information gain, is just one of the possible approaches.  In theory, any method that ensures the total price to be paid to all the federations, in exchange of their data, is within $D$'s budget, and the privacy parameters proposed are within the corresponding privacy budgets of the federations, could be implemented to propose a revised set of privacy parameters and, in turn, the price associated with them. 
    
    \begin{definition}[Seal the deal]
       When all the federations are informed about the revised privacy parameters desired of them, and they agree to proceed with the private-data trading with the data consumer by achieving the revised privacy parameter by combining the data of their members, we say \emph{the deal has been sealed} between the federations and the data consumer.
    \end{definition}
    
    Once the deal is sealed between the federations and the data consumer, $F_i$ is expected to provide data gathered from its members with an overall obfuscation with the privacy parameter $\epsilon^P_{F_i}$ using the $k$RR mechanism, in exchange of a price  $M^i=f^{-1}(\epsilon^P_{F_i})$ for every $i \in \{1,\ldots,k\}$. Failing to achieve this parameter of privacy for any federation results in a failure to uphold the conditions of the ``deal'' and makes the deal  void for that federation, with no price received. 
    
    A rational assumption made here is that if a certain federation $F$ fails to gather data from its members such that the overall $k$RR privacy parameter of $F$ is less than $\epsilon^P_F$, then $F$ doesn't receive any partial compensation for its contribution, as it would incur an increase in communication cost and time for the data consumer in proceeding to this stage and ``seal a new deal'' with $F$, instead of investing the revenue to a more responsible federation.
    
    The rest of the process consists in collecting the data and it takes place within 
    every federation  $F$ which has sealed the deal. At the $t^{th}$ round, for $t\in \{1,2,\ldots\}$, any member $p$   of $F$   has the freedom of contributing $d^{t}_p \leq d_p - \sum_{i=1}^{t-1}d^{i}_p$ data points privatized using $k$RR mechanism with any parameter $\epsilon^t_p$. The process continues until the overall information collected until then achieves a privacy level of of at least $\epsilon^P_F$. Let $\cal T$ denote the number of rounds needed by $F$ to achieve the required privacy level. As per the deal sealed between $F$ and $D$,  $F$ needs to submit $D_F=\sum_{p \in F}\sum_{i=1}^{\cal T}d^i_p$ data points to $D$ such that the overall $k$RR privacy level of the collated data, $\epsilon_F$, is at least $\epsilon^P_F$, and in return $F$ receives a financial revenue of $\$M$ from $D$.

    \subsection{Earning Splitting}
   We  use the Shapley value to estimate the contribution of each data provider of the federation, in order to split the whole earning $M$, which $F$ would receive from $D$ at the end of the trade. Let $\psi:\mathbb{R}^+\times\mathbb{R}^+ \mapsto \mathbb{R}^+$ be the valuation function used for evaluating the Shapley values of the members after each contribution. If a certain member, $p$, of $F$ reports $d$ differentially private data points with privacy parameter $\epsilon$, $\psi_i(v)$ should give the share of ``contribution'' made by $p$ over the total budget, $M$, of $F$, to be split across all its members. It is assumed that each member, $p$, of $F$ computes her Shapley value, knows what share of revenue she would receive by contributing her data privatized with a chosen privacy parameter, and uses this knowledge to decide on $\epsilon^t_p$ at every round $t$, depending on her financial desire. In our model, characteristic function $v(S)$ is as follows:
    \[ v(S)=
\begin{cases}
M, & \mbox{if }	 \epsilon_F \geq \epsilon_{F}^P \\
0, & \mbox{if }  \epsilon_F < \epsilon_{F}^P
\end{cases} \]
where $n$ is the number of data provider in subset $S$ .

\begin{example}
As an example, let us assume that there are  $p_1$, $p_2$, $p_3$, and each provider's contribution $\sum_{t=1}^{\cal T}d^t_p\frac{e^{\epsilon^t_{p}}}{k-1 + e^{\epsilon^t_{p}}}$ are $1.0$, $0.5$ and $0.3$. And we assume that $\epsilon_F^P$ is 1.4 and financial revenue of $M$ is 60. In this case, the calculation of each provider's revenue using Shapley value is as follows:\\
\\
Case 1) Only one data provider participates:
\begin{center}
$p_1: v(p_1)=0$\\
$p_2: v(p_2)=0$\\
$p_3: v(p_3)=0$ 
\\
\end{center}
Case 2) Two providers participate: $v(p_1+)$=0,$v(p_2)$=0,
\begin{center}
$p_1: v(p_1+p_2)-v(p_2)=M,  v(p_1+p_3)-v(p_3)=M$ \\
$p_2: v(p_1+p_2)-v(p_1)=M,  v(p_2+p_3)-v(p_3)=0$\\
$p_3: v(p_1+p_3)-v(p_1)=0,  v(p_2+p_3)-v(p_2)=0$
\\
\end{center}
Case 3) All providers participate:
\begin{center}
$p_1: v(p_1+p_2+p_3)-v(p_2+p_3)=M$ \\
$p_2: v(p_1+p_2+p_3)-v(p_1+p_3)=M$ \\
$p_3: v(p_1+p_2+p_3)-v(p_1+p_2)=0$ \\

\end{center}
According to the above results, the share of each user, according to their Shapley values, is as follows:
\begin{center}
$\psi_1(v)=\frac{0!2!}{3!}0+\frac{1!1!}{3!}M+\frac{1!1!}{3!}M+\frac{2!0!}{3!}M=\frac{4M}{6}$=40\\
$\psi_2(v)=\frac{0!2!}{3!}0+\frac{1!1!}{3!}M+\frac{1!1!}{3!}0+\frac{2!0!}{3!}M=\frac{2M}{6}$=20\\
$\psi_3(v)=\frac{0!2!}{3!}0+\frac{1!1!}{3!}0+\frac{1!1!}{3!}0+\frac{2!0!}{3!}0=\frac{0M}{6}$=0\\
\end{center}

In this example, $p_3$ has no effect on achieving the $\epsilon_F^P$ . Thus, $p_3$ is excluded from the revenue distribution. If the revenue were distributed proportionally, without considering the Shapley values, the revenue of $p_1$ would be 33, $p_2$ is 17, and $p_3$ is 10. It would mean $p_1$ and $p_2$ would receive lower revenues even though their contribution are sufficient to achieve the $\epsilon_F^P$, irrespective of the participation of $p_3$. The Shapley value enables the distribution of revenues only for those who have contributed to achieving the goal.
\end{example}

One of the problems of computing the Shapley values is the high computational complexity involved. If there is a large number of players, i.e., the size of a federation is large, the total number of subsets to be considered becomes considerably large, engendering a limitation to real-world applications. To overcome this, we use a \emph{pruning technique} to reduce the computational complexity of the mechanism. A given federation $F$ receives revenue $M$ only when
$\epsilon_F \geq \epsilon_F^P$, as per the deal sealed with the data consumer. Therefore, it is not necessary to calculate for Shapley values for the cases where $\epsilon_F < \epsilon^P_F$, since such cases do not contribute towards the overall Shapley value evaluated for the members of $F$.  

It is reasonable to assume this differentially private data trading between the data consumer and the federations would continue periodically for a length of time. For example, Acxiom, a data broker company, periodically collects and manages personal data related to daily life, such as consumption patterns and occupations. Periodic data collection has higher value than one-time data collection because it can 
track temporal trends. For simplicity of explanation, let's assume that the trading occurs ever year. Hence, we consider a yearly period to illustrate the final two steps of our proposed mechanism - ``swift data collection'' and the ``penalty scheme''. This would ensure that the data collection process is as quick as possible for every federation in every year. Additionally, this would motivate the members to cooperate and act in the best interests of their respective federations by not, unnecessarily, withholding their privacy contributions to hinder achieving the privacy goals of their group, as per the deal finalized with $D$.

Let $R\in \mathbb{N}$ be the ``tolerance period''. For a member $p \in F$, we denote $d(m)^i_p$ to be the number of data points reported by $p$ in the $i^{th}$ round of data collection of year $m$ and we denote $\epsilon(m)^i_p$ to be the privacy parameter used by $p$ to obfuscate the data points in the $i^{th}$ round of data collection of year $m$. Let $T_m$ be the number of rounds of data collection needed in year $m$ by federation $F$ to achieve their privacy goal. We denote the total number of data points reported by $p$ in the year $m$ by $d(m)_p$, and observe that $d(m)_p=\sum_{i=1}^{T_m}d(m)^i_p$. Let $\epsilon(m)^P$ denote the value of the privacy parameter of the combined $k$RR mechanism of the collated data that $F$ needs, in order to successfully uphold the condition of the deal sealed with $D$.

\begin{definition}[Contributed privacy level]
For a given member $p\in F$, we define the \emph{contributed privacy level} of $p$ in year $m$ as
\[\epsilon(m)_p=\sum\epsilon(m)^i_p\].

\end{definition}

\begin{definition}[Privacy saving] For a given member $p\in F$, we define the \emph{privacy saving} of $p$ over a tolerance period $R$ (given by a set of some previous years), decided by the federation $F$, as
\[\Delta_p = \sum_{m\in R}\left(d(m)_p\epsilon^T_p-d(m)_p\epsilon(m)_p\right) \]
\end{definition}

\paragraph{Swift data collection:}
It is in the best interest of $F$, and all its members, to reduce the communication cost, time, and resources over the data collection rounds, and achieve the goal of $\epsilon^P$ as soon as possible, to catalyze the trade with $D$, and receive the financial revenue. We aim to capture this through our mechanism, and enable the members not to ``hold back'' their data well below their capacity.

To do this, in our model we design the Shapley valuation function, $\psi(.)$, such  that for $p\in F$, in year $m$, $\psi(N_p\epsilon(m)^{t+1}_p,d(m)_p,M)$ $=\psi(\epsilon(m)^{t}_p,d(m)_p,M)$, where $N_p \in \mathbb{Z}^+$ is the \emph{catalyzing parameter} of the data collection, decided by the federation, directly proportional to $\Delta_p$. In particular, for $p\in F$, and a tolerance period $R$ decided, in prior, by $F$, it is a reasonable approach to make $N_p\propto\Delta_p$, as this would mean that any member $p\in F$, reporting $d(m)_p$ data points, would need to use $N_p$ times higher value of $\epsilon$ in the $(t+1)^{st}$ round of data collection in the year $m$, as compared to that in the $t^{th}$ round for the same number of data points reported to get the same share of the benefit of the federation's overall revenue, where $N_p$ is decided by how much privacy savings $p$ has had over a fixed period of $R$. 

This is made to ensure that if a member of a federation has been holding back her information by using high values of privacy parameters over a period of time, she should need to compensate in the following year by helping to quicken up the process of data collection of her federation. This should motivate the members of $F$ to report their data with a high value of the privacy parameter in earlier rounds than later, staying within their privacy budgets, so that the number of rounds needed to achieve $\epsilon(m)^P$ is reduced.

\paragraph{Penalty scheme:} It is also desirable to have every member of any given federation to cooperate with the other members of the same federation, and facilitate the trading process in the best interest of the federation, to the best of their ability. That is why, in our mechanism, we incorporate an idea of a ``penalty scheme'' for the members of a federation who are being selfish by keeping a substantial gap between their maximum privacy threshold and their contributed privacy level,  wishing to enjoy benefits of the revenue at an unfair cost of other members providing information privatized with almost their maximum privacy threshold. To prevent such non-cooperation and attempted ``free ride'', we design a ``penalty scheme'' in the mechanism.

\begin{definition}[Free rider]
We call a certain member $p\in F$ to be a \emph{free rider} if $\Delta_p \geq \delta_F$, for some $\delta_F \in \mathbb{R}^+$. Here, $\delta_F$ is a threshold decided by the federation $F$ beforehand and informed to all the members of $F$.  
\end{definition}  

Thus, in the ideal case, every member of $F$ would have their privacy savings to be 0 if everyone contributed information to the best of their abilities, i.e., provided data obfuscated with their maximum privacy parameter.  But as a federation, a threshold amount of privacy savings is tolerated for every member. Under the ``penalty scheme'', if a certain member $p \in F$ qualifies as a free rider, she is excluded from the federation, and is given a demerit point by the federation, that can be recorded by a central system keeping a track of every member of every federation, preventing $p$ from getting admission to any other federation for her tendency to free ride. This would mean $p$ and has the responsibility of trading with the data consumer by herself. We could define the Shapley valuation function used to determine the share of $p$'s contribution such that $f^{-1}(\epsilon^T_p)< \psi(v,M)$, implying that the revenue to be received by $p$ dealing directly with $D$, providing one data point obfuscated with her maximum privacy threshold with respect to the privacy valuation function $f(.)$, would be giving a much lower revenue than what $p$ would receive being a member of federation $F$. \footnote{Here, $v(.)$ is the characteristic function of $\psi(.)$, depending on $\epsilon_p^T$.}

\begin{restatable}{theorem}{penaltyscheme}\label{th:1}
If the privacy valuation function used by the data consumer, $D$, is $f(m)=K_1(e^{K_2m}-1)$, in order to impose the \emph{penalty scheme} to any member $p\in F$ of a federation $F$, the Shapley valuation function, $\psi(.)$, chosen by $F$, must satisfy $\frac{\ln(\frac{\epsilon^T_p}{K_1}+1)}{K_2} < \psi\left(\epsilon^T_p, \frac{\ln(\frac{w^*\epsilon^T_p}{K_1}+K)}{K_2}\right)$, where $K=\frac{\sum_{p'\neq p \in F}d_{p'}\epsilon^T_{p'}}{K_1}+1$, $d_{\pi}$ is the number of data points reported by any $\pi\in F$, and $w^*$ is the suggested scaling parameter computed by $D$ to propose a realistic deal, as described in section 4.1.
\end{restatable}

\begin{proof}
See Appendix~\ref{app:a}
\end{proof}

Imposing the ``penalty scheme'' is expected to drive every member of a given federation to be cooperating with the interests of the federation and all the other fellow members to the best of their abilities, preventing potential free riders.

 We show the pseudocode for the entire process in Algorithm \ref{alg:entireAlg} and describe the swift data collection and penalty scheme in Algorithm \ref{alg:swiftAlg} and \ref{alg:penaltyAlg}. 

\begin{algorithm}
\SetAlFnt{\small}
\textbf{Input:} Federation $F$, Data consumer $D$\;
\textbf{Output:} $\epsilon^P_F$ and $M$\;
$D$ broadcasts total budget $B$ and $f(.)$\;
Federation $F$ computes the $\epsilon^T_F=\sum_{i=1}^nd_{p_i}\epsilon^T_{p_i}$\;
$p_*$ places a bid to D to obtain revenue $M$\;
$F$ and $D$ ``seal the deal'' to determine the $\epsilon^P_F$ and $M$\;
\While{$\epsilon_F \leq \epsilon^P_F$ and $t \leq T$}{
\Call{SwiftDataCollection}{$F$, $\epsilon^P_F$}\;
 $p_*$ computes the overall privacy $\epsilon_F$
}
\eIf{$\epsilon_F \geq \epsilon^P_{F_i}$}
{
 F receives the revenue $M$\;
 $p_*$ computes the Shapley value $\psi_i(v, M)$\;
 $p_i$ get their share of the revenue $M$
 }
{
 deal fails
}  
\caption{Federation based data trading algorithm}\label{alg:entireAlg}
\end{algorithm}

\begin{algorithm}
\SetAlFnt{\small}
\textbf{Input:} $F=\{p_1,\ldots,p_{n_F}\}$,$\epsilon^P_F$\;
\textbf{Output:} $\epsilon(m)^t_p$\;
\SetKwFunction{Fmain}{SwiftDataCollection}
\SetKwProg{Fn}{Function}{:}{}
\Fn{\Fmain{$F$,$\epsilon^P_F$}}{ 
\While{$i \leq n_F$}
    {
        Compute $\Delta_{p_i}$\;
        Compute the catalyzing parameter $N_{p_i}$\;
        Determine the $\epsilon(m)_{p_i}^{t}=N_{p_i}\epsilon(m)_{p_i}^{t-1}$
    }
}

\caption{Swift data collection algorithm}\label{alg:swiftAlg}
\end{algorithm}

\begin{algorithm}
\SetAlFnt{\small}
\textbf{Input:} $F=\{p_1,\ldots,p_{n_F}\}$,$\Delta_{F}=\{\Delta_{p_1},\ldots ,\Delta_{p_{n_F}}\}$, $\delta_F$\;
\textbf{Output:} Updated $F$\;
\While{$i \leq n_F$}
{
    \If{$\Delta_{p_i} \geq \delta_F$}
    {
        $ F\setminus \{i\}$
    }
}
\caption{Penalty scheme}\label{alg:penaltyAlg}
\end{algorithm}

\section{Experimental results}

\subsection{Experimental environments}

In this section, we show some  experiments that support the claim that proposed method succeeds to obtain the promised $\epsilon$ and reduce the computation time for Shapley value evaluation. The number of data providers constituting the federation is set to 25, 50, 75, and 100, respectively. The value of $\epsilon^T_p$ is selected from the normal distribution between 1 and 10 with mean 5 and standard deviation 1 independently for all participants $p$ in the federation. The experimental environment is a Intel(R) i5-9400H CPU and 16 GB of memory.

\subsection{Number of rounds needed for data collection}
\begin{figure}[htbp]
\centerline{\includegraphics[width=1\textwidth]{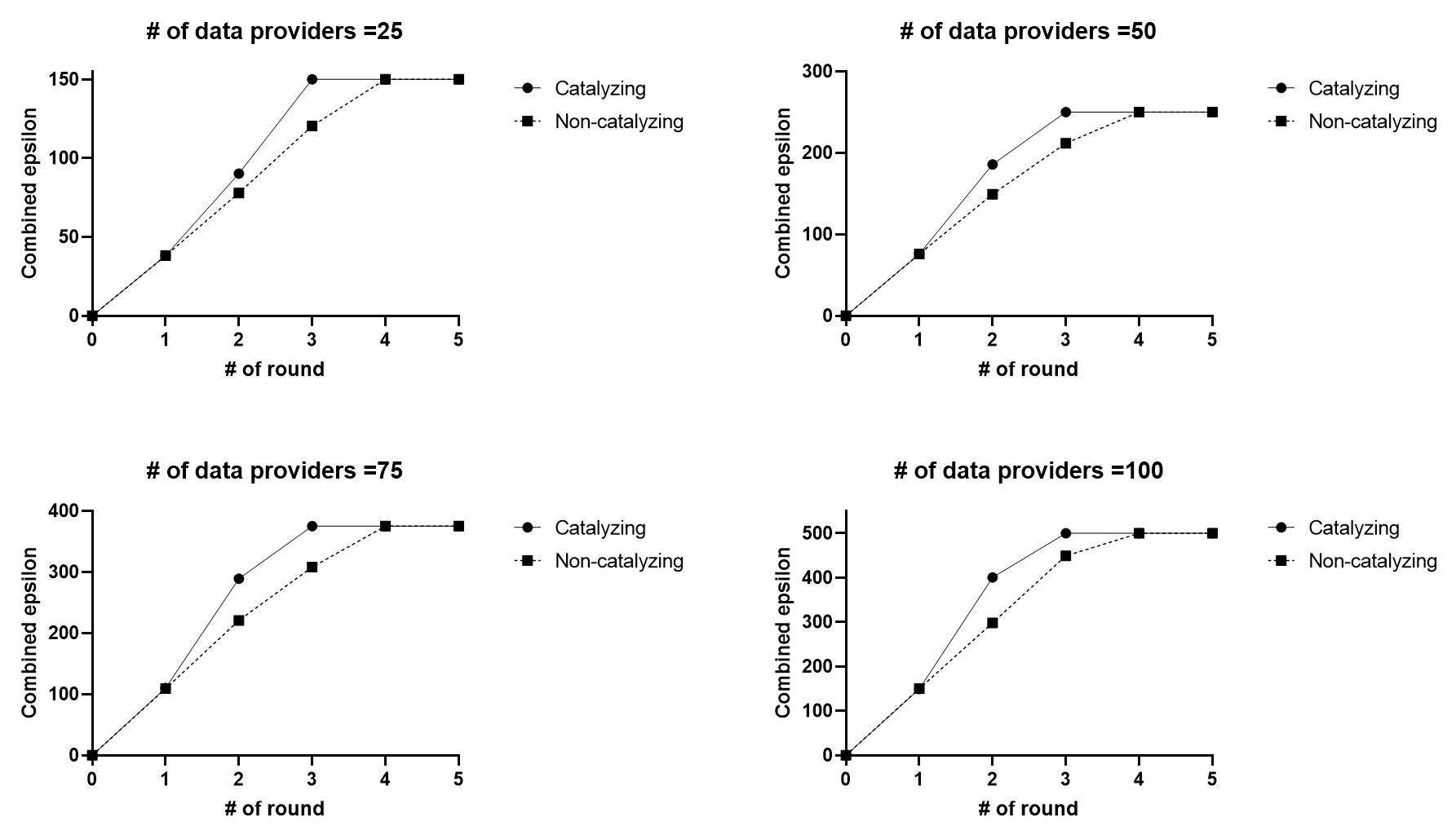}}
\caption{Experimental results for combined $\epsilon $. Combined $\epsilon$ refers to the amount of information provided by the data providers. }
\label{fig2}
\end{figure}
Achieving the $\epsilon_F^P$ is the key for the participation of $F$ in the  data trading. If $\epsilon_F^P$ is not achieved as the collated information level for the federation, there is no revenue from the data consumer. Thus, it is important to encourage the data providers to report sufficient data in order to reach the goal of the deal sealed with the data consumer. The swift data collection is a way to catalyze the process of obtaining data from the members of every federation $F$, minimising the number of rounds of data-collection, to achieve $\epsilon_F^P$.  Furthermore, we set  $N_p=\frac{\Delta_p}{d(m)_p\epsilon^T_p}$ for a certain member $p$ in federation $F$, to motivate the data providers who have larger  privacy savings to provide more information per round. 

In the experiment, $\epsilon_F^P$ is set to be 125, 250, 375 and 500, respectively. Data provider $p$ determines $\epsilon(m)^t_p$ randomly in first round, and then computes $\epsilon(m)^t_p$ according to $N_p$, for every $p$ in the federation. We compare  two cases, the catalyzing method and the non-catalyzing method.

As illustrated in figure \ref{fig2}, the experimental results show that both catalyzing data collection and its non-catalyzing counterpart achieve the promised epsilon values within 5 rounds, but it can be seen that the catalyzing method achieves $\epsilon_F^P$ in an earlier round because data providers decide the privacy level used to obfuscate their data with, considering their privacy savings, resulting in a \emph{swift data collection}.

\subsection{Number of free riders by penalty scheme}

\begin{figure}[htbp]
\centerline{\includegraphics[width=1\textwidth]{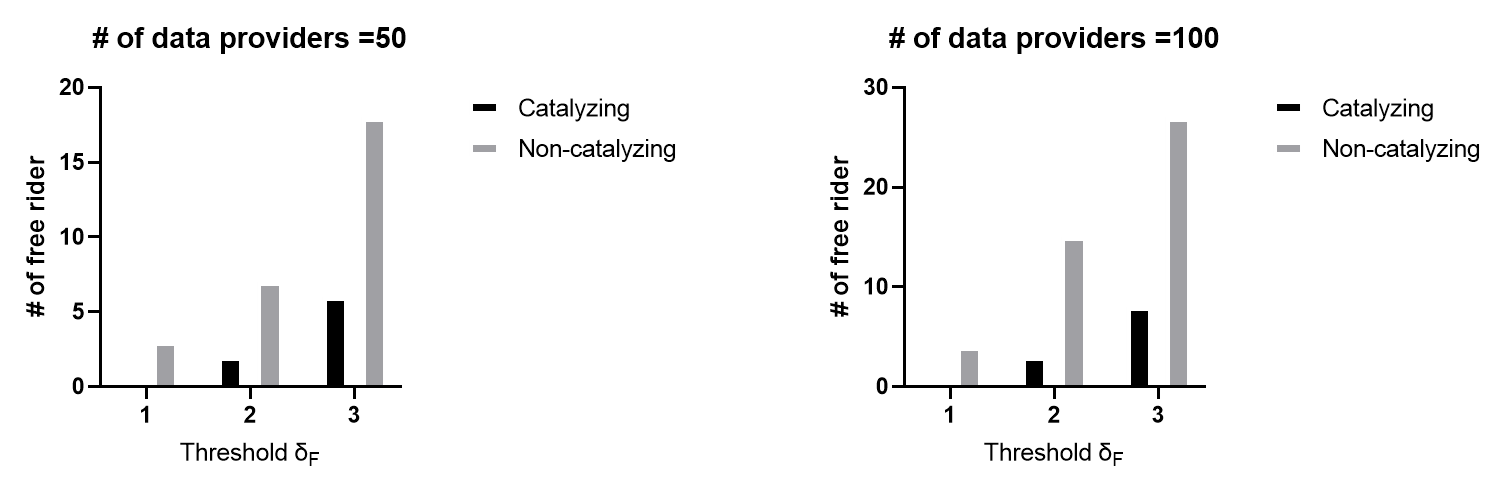}}
\caption{Experimental results for number of free riders. We compared the number of free riders incurred by the penalty scheme in catalyzing and non-catalyzing methods for cases where the number of data providers is 50 and 100. }
\label{fig3}
\end{figure}

The penalty scheme that prevents free riders is based on the premise that trading data by participating in a federation is more beneficial than trading data directly with data consumers (Theorem \ref{th:1}). We evaluated the number of free riders in the catalyzing and non-catalyzing methods according to the increase of the threshold $\delta_F$ in the experiment.

As shown in the figure \ref{fig3}, we can see that the number of free riders increases in both techniques as the threshold value $\delta_F$ is increased to 1,2,3. However, the non-catalyzing method makes more free riders than the catalyzing method that changes the amount of provided information according to privacy saving $\Delta_P$. In other words, the catalyzing method and penalty scheme help to keep members in the federation by inducing them to reach the target epsilon in an earlier time.

\subsection{Reduced Shapley value computation time }
As mentioned in section 4.2, one of the limitations of Shapley value evaluation is to compute it for all combinations of subsets. Through this experiment, we demonstrate that the proposed pruning technique reduces the computation time for calculating the Shapley values. We compared the computation times of the proposed method with brute force method that calculates all the cases by increasing the number of data providers in the federation, by 3, from 15 to 27.

\begin{table}
\centering
\caption{Computation time of brute force and proposed pruning method}\label{tab1}
\begin{tabular}{|c|c|c|}
\hline
\bfseries \# of data providers&  \bfseries brute force(Sec) & \bfseries pruning method (Sec)\\
\hline
15 &  0.003 & 0.0007\\
18&  0.02 & 0.001\\
21 & 0.257 & 0.0049\\
24 & 2.313 & 0.009\\
27 & 19.706 & 0.019\\
\hline
\end{tabular}
\end{table}
As shown in the table, the computation time of Shapley value evaluation increases exponentially because the total number of subsets to be considered does the same. The proposed method can calculate the Shapley values in less time by removing unnecessary computations.

\section{Conclusion}
With the spreading  of data-driven decision making practices, the interest in personal data is increasing. The data market gives a new opportunity to trade personal data, but a lot of research is still needed to solve privacy and pricing issues. 
In this paper, we have considered a data market environment in which data providers form federations and protect their data with the locally differentially private $k$RR mechanism, and we have proposed a pricing and earnings-distribution method. Our method integrates different data providers' values of the privacy parameter $\epsilon$ and combines them to obtain the  privacy parameter  of the federation. The received earning is distributed using the Shapley values of the members, which guarantees the \emph{Pareto efficiency} and \emph{symmetry}. In addition, we have proposed a swift data collection mechanism and a penalty scheme to catalyze the process of achieving the target amount of information quickly, by penalizing the free riders who do not cooperate with their federation's best interest.

Our study has also disclosed new problems that need further investigation. Firstly, we are assuming that the data providers keep the promise for the ``seal the deal", but, in reality, the data providers can always add more noise than what they promised. We plan to study how to ensure that data providers uphold their data trading contracts. Another direction for future work is considering more  differential privacy mechanisms, other than $k$RR. 

\section*{Acknowledgement}
This work was supported by the European Research Council (ERC) project HYPATIA under the European Union’s Horizon 2020 research and innovation programme. Grant agreement n. 835294.

 \bibliographystyle{splncs04}
 \bibliography{reference}
\begin{subappendices}
\renewcommand{\thesection}{\Alph{section}}%
% or try \arabic{section}

\section{Proofs}\label{app:a} 

\penaltyscheme*
\begin{proof}
Using the privacy valuation function $f(m)=K_1(e^{K_2m}-1)$, we have $f^{-1}(\epsilon)=\frac{\ln(\frac{\epsilon}{K_1}+1)}{K_2}$. Let $p$ be an arbitrary member of $F$ with a maximum privacy threshold $\epsilon^T_p$. Therefore, in order to impose a penalty scheme on $p$, it needs to be ensured that
\begin{flalign}
    \frac{\ln(\frac{\epsilon^T_p}{K_1}+1)}{K_2} < \psi(v, M) \nonumber&&&\\\nonumber
    &\\\nonumber\implies
    \frac{\ln(\frac{\epsilon^T_p}{K_1}+1)}{K_2} < \psi(v, f^{-1}(\epsilon^P_F))&&\\\nonumber\nonumber
    &\\\nonumber
    [w^*\in[0,1]\text{ is the scaling parameter chosen by $D$ and } \epsilon^P_F=w^*\epsilon^T_F]&\\\nonumber\nonumber
    &\\\nonumber\implies 
    \frac{\ln(\frac{\epsilon^T_p}{K_1}+1)}{K_2} < \psi\left(v, \frac{\ln(\frac{\epsilon^P_F}{K_1}+1)}{K_2}\right)&&\\\nonumber
    &\\\nonumber\implies 
    \frac{\ln(\frac{\epsilon^T_p}{K_1}+1)}{K_2} < \psi\left(v, \frac{\ln(\frac{C_0+w^*\epsilon^T_p}{K_1}+1)}{K_2}\right)&&\\\nonumber
    &\\ \nonumber
    [\text{ where }C_0=\sum_{p'\neq p \in F}d_p'\epsilon^T_{p'}\text{ is a constant}]&\\\nonumber\nonumber
    &\\ \nonumber \implies
    \frac{\ln(\frac{\epsilon^T_p}{K_1}+1)}{K_2} < \psi\left(v, \frac{\ln(\frac{C_0+w^*\epsilon^T_p}{K_1}+1)}{K_2}\right)&&\\\nonumber
    &\\ \nonumber
    \frac{\ln(\frac{\epsilon^T_p}{K_1}+1)}{K_2} < \psi\left(v, \frac{\ln(\frac{w^*\epsilon^T_p}{K_1}+K)}{K_2}\right) \label{penaltycondi}&&\\
    &\\ \nonumber
    [\text{for the constant }K=\frac{C_0}{K_1}+1.]
\end{flalign}

\end{proof}

\end{subappendices}
\end{document}